\documentclass[journal ]{new-aiaa}[2018/01/10, v1.2]

\title{Approximate Stochastic Optimal Control for Linear Time Invariant Systems with Heavy-tailed Disturbances}

\author{Shawn Priore 
\footnote{Graduate student, Department of Electrical and Computer Engineering} 
and Meeko Oishi.
\footnote{Professor, Department of Electrical and Computer Engineering}}
\affil{University of New Mexico, Albuquerque, NM. 87131}
\author{Christopher Petersen 
\footnote{Assistant Professor, Department of Mechanical and Aerospace Engineering}}
\affil{University of Florida, Gainesville, FL. 32611}

\usepackage[utf8]{inputenc}
\usepackage{textcomp}
\usepackage{graphicx}
\usepackage{amsmath}
\usepackage[version=4]{mhchem}
\usepackage{siunitx}
\usepackage{longtable,tabularx}
\usepackage{url}
\usepackage{bm}
\usepackage{multirow}
\usepackage{cuted}
\usepackage{mathrsfs}

\usepackage{amsthm}

\newtheorem{defn}{Definition}
\newtheorem{prop}{Property}
\newtheorem{lem}{Lemma}

\newtheorem{prob}{Problem}
\newtheorem{subprob}{Problem}[prob]
\newtheorem{assm}{Assumption}

\usepackage{algorithm}
\usepackage{algpseudocode}

\algdef{SE}[DOWHILE]{Do}{doWhile}{\algorithmicdo}[1]{\algorithmicwhile\ #1}%

\newcommand{\Prob}{\mathbb{P}}
\newcommand{\R}{\mathbb{R}}
\newcommand{\N}{\mathbb{N}}
\newcommand{\cov}{\mathrm{cov}}
\newcommand{\bvec}[1]{\vec{\boldsymbol{#1}}}

\begin{document}
\maketitle

\begin{abstract}
    We propose an open loop control scheme for linear time invariant systems perturbed by multivariate $t$ disturbances through the use of quantile reformulations. The multivariate $t$ disturbance is motivated by heavy tailed phenomena that arise in multi-vehicle planning planning problems through unmodeled perturbation forces, linearization effects, or faulty actuators. Our approach relies on convex quantile reformulations of the polytopic target sets and norm based collision avoidance constraints to enable fast computation. We embed quantile approximations of the Student's $t$ distribution and the beta prime distribution in a difference-of-convex function framework to compute provably safe but likely suboptimal controllers. We demonstrate our method with three satellite rendezvous examples and provide a comparison with particle control. 
\end{abstract}

\section*{Nomenclature}

{\renewcommand\arraystretch{1.0}
\noindent\begin{longtable*}{@{}l @{\quad=\quad} l@{}}
$F_x, F_y, F_z$ & force in the $x$, $y$, and $z$ direction, N \\
$F_{\theta}$ & rotational force in the $\theta$ direction, N$\cdot$m \\
$\overline{I}_n$ & $n$-dimensional identity matrix \\
$J_{\theta}$ & moment of inertia, kg$\cdot$m\textsuperscript{2} \\
$k$ & discrete time index \\
$m_c$ & mass, kg \\
$N$ & final time index \\
$R_0$ & orbital radius, m \\
$s_i, \eta_i, \upsilon_i$ & slack variables \\
$\mathscr{T}$ & polytopic target set \\
$\vec{u}(k)$ & input at time step $k$ \\
$\vec{U}$ & concatenated input vector\\
$\mathscr{U}$ & admissible input space \\
$v$ & number of satellites \\
$\bvec{w}(k)$ & disturbance vector at time $k$\\
$\bvec{W}$ & concatenated disturbance vector \\
$\bvec{x}(k)$ & state at time step $k$ \\
$\bvec{X}$ & concatenated state vector \\
$\mathscr{X}$ & admissible state space\\ 
$BPrime(\gamma, \delta)$ & beta prime distribution with shape parameters $\gamma$ and $\delta$\\
$CSquare(\nu)$ & chi-square distribution with parameter $\nu$ \\
$Gamma(\cdot, \cdot)$ & gamma distribution \\
$Normal(\vec{\mu}, \overline{\Sigma})$ & multivariate Gaussian distribution with parameters $\vec{\mu}$ and $\overline{\Sigma}$\\
$t(\vec{\mu}, \overline{\Sigma}, \nu)$ & multivariate $t$ distribution with parameters $\vec{\mu}$, $\overline{\Sigma}$, and $\nu$\\
$\alpha$ & probabilistic violation threshold \\
$\beta (\cdot,\cdot)$ & beta function \\
$\Gamma(\cdot)$ & gamma function \\
$I_{x}(\cdot, \cdot)$ & incomplete beta function evaluated at $x$ \\
$\lambda_{\max}(\cdot)$ & maximum eigenvalue of the input matrix \\
$\mu$ & gravitational constant, m${}^{3}$ kg${}^{-1}$ s${}^{-2}$ \\
$\vec{\mu}$ & location vector parameter \\
$\nu$ & degrees of freedom parameter \\
$\overline{\Sigma}$ & scale matrix parameter \\
$\phi(\cdot)$ & probability density function (pdf) \\
$\Phi(\cdot)$ & cumulative density function (cdf)\\
$\Phi^{-1}(\cdot)$ & quantile function \\
$\omega$ & orbital rate  \\
$\overline{0}_{n\times m}$ & $n\times m$ dimensional matrix of zeros \\\
$\N$ & natural numbers with 0\\
$\N_{N}$ & natural numbers from 0 to N inclusive\\
$\N^{+}$ & natural numbers without 0\\
$\N^{+}_{N}$ & natural numbers from 1 to N inclusive \\
$\R_+$ & positive real numbers \\
$\R^{n \times n}_{++}$ & set of $n \times n$ positive definite matrices \\
$\mathrm{diag}(\overline{A}_1, \dots,\overline{A}_q)$ & a block diagonal matrix with elements $\overline{A}_1,\dots,\overline{A}_q$  \\
$\|\vec{\cdot}\|$ & 2-norm of an input vector \\
\multicolumn{2}{@{}l}{Subscripts}\\
$i-j$ & difference of the $i$ and $j$\textsuperscript{th} vehicle's vectors
\end{longtable*}}

\section{Introduction} \label{sec:intro}

Enabling extended satellite lifetimes through advances in on-orbit refueling and servicing depots has been a focal point for many private and public organizations. Such advances are requiring new technologies to enable efficient autonomous coordination between multiple satellites despite the harsh environment and limited resources such as fuel or computational abilities. These new technologies must accommodate path planning and optimization for mission critical vehicles under uncertain conditions that may arise from modeling inaccuracies, inaccurate or time delayed sensing, and faulty or inconsistent actuation mechanisms. These stochastic elements are frequently modeled using Gaussian disturbances for mathematical convenience. However, noise processes from conditions like these can take on non-Gaussian characteristics, such as heavy tailed phenomena. Heavy tailed distributions are defined as distributions with tail probabilities do not exhibit exponential decay. Another interpretation would be that outlier events are more common. In stochastic systems, heavy tailed phenomena can arise from faulty thrusters or sensors, extreme weather or geological events such as hurricanes or earthquakes \cite{Pisarenko2010}, or even magnetic disturbances caused by solar winds and solar flares \cite{Popescu2010}. Computation of controllers that meet required probabilistic safety thresholds for target acquisition and collision avoidance in these conditions require accurate assessments of disturbance characteristics. In this paper we seek to construct solutions for stochastic optimal control problems in an computationally efficient and tractable manner for cooperative multi-vehicle coordination problems exhibiting heavy tailed noise processes. 

Stochastic problems often result in intractable integrals that require techniques that trade off optimally for tractability. Moment based approaches  \cite{nemirovski2006convex, calafiore2006distributionally, paulson2017stochastic} require analytic expressions for computational parameters that may only exist conditionally \cite{Idan2008, Duong2019}, or introduce conservative reformulations of constraints through Boole's inequality \cite{oldewurtel2014stochastic,ono2008iterative,vitus_feedback_2011}. Fourier transforms have been used to bypass the quadrature computation required to evaluate probability integrals \cite{Idan2019}, and in combination with piecewise affine approximations, have been used to evaluate chance constraints \cite{Sivaramakrishnan2021TAC, vinod2019piecewise} for linear time-invariant (LTI) systems with noise processes that elicit log-concave probability density functions (pdf). Sample based approaches have been employed for systems with known disturbances \cite{calafiore2006scenario,blackmore2011chance}.  Sample reduction techniques have been employed to increase accuracy and efficiency for both convex \cite{campi2011sampling,care2014fast} and non-convex \cite{Campi2018TAC} problems. However, limited computational capacity continues to hindered accuracy as these method can only guarantee safety asymptotically \cite{calafiore2006scenario,blackmore2011chance}. Learning based techniques are standard for dealing with heavy tailed disturbances \cite{Idan2008, Li2014, Wang2008} but are reliant on computationally intense processes and may lack in interpretability. 

With the near exponential increase in active satellites and debris over the past few years, collision avoidance has become an increasingly important consideration in satellite control. Sequential convex programming can be applied to saturation penalty functions based on 2-norm collision avoidance constraints but suffers from singularities and potential non-convergence during gradient decent \cite{Zhao2022}. Techniques for calculating collision avoidance probabilities have been posed \cite{Enriquez2022}, but are difficult to embed in control algorithms and may be limited by the shape of the satellite. Methods for rendezvous and proximity operations between a controlled satellite and a stationary or potentially non-cooperative satellite have been proposed \cite{DiMatteo2012, Maestrini2022, Ulybyshev2011} but lack in their ability to accommodate multiple controlled vehicles. Previous work of ours \cite{PrioreACC21} has solved this problem in Gaussian environments with difference-of-convex function optimization but reliance on predetermining risk allocation between constraints led to sub-optimal solutions, and the method cannot accommodate non-Gaussian disturbances. 

The work presented in this paper uses the theory presented in \cite{prioreACC22} for arbitrary disturbances to extend the work presented in \cite{PrioreACC21} to accommodate the both a multivariate $t$ disturbance and embed risk allocation in the optimization problem. This approach employs the use of quantile reformulations and approximations to solve chance constrained stochastic optimization problems. As in \cite{PrioreACC21, prioreACC22}, we consider cooperative multi-vehicle planning problems with two types of constraints: a) polytopic target set constraints, and b) norm-based collision avoidance constraints. Constraints like these arise where multiple vehicles must reach desirable states while maintaining safe distances from other vehicles and static objects in the environment. We show that constraints of this form can be bounded by constraints that are affine in the control input and disturbance. The target set constraints yield linear constraints, while the collision avoidance constraints are reverse-convex. These bounds are conservative with respect to the initial chance constraints and result in a tightened approximation of the original problem. Once in this form, satisfying the probabilistic chance-constraints no longer require evaluating a series of multidimensional integrals. However, evaluation of the quantile, the inverse of the cumulative distribution function (cdf), is required to evaluate the chance constraints within this new formulation. For both constraint types, the quantile does not elicit an analytic form. 

The theory presented in \cite{prioreACC22} employs a Taylor series approximation of the quantile function based on successive derivatives of the pdf. We evaluate the Taylor series at regular intervals to generate a piecewise affine approximation of the quantile. This enables us to embed the quantile in a difference-of-convex programming framework \cite{boyd_dc_2016}.  We employ an iterative approach, as in \cite{ono2008iterative, PrioreACC21, prioreACC22}, to compute optimal controllers and allocate risk. This approach exploits convexity to enable fast computation despite being iterative. We note that probabilistic guarantees can only be satisfied asymptotically given the present but unknown error in the quantile approximation. {\em The main contribution of this paper is to develop a difference-of-convex framework that enables efficient evaluation of target set and collision avoidance probabilistic chance-constraints in cooperative multi-vehicle planning problems for LTI systems with heavy tailed noise processes.} 

The paper is organized as follows. Section \ref{sec:prelim} provides mathematical preliminaries and formulates the optimization problem. Section \ref{sec:methods} reformulates the chance constraints by approximating the quantile function.  Section \ref{sec:results} demonstrates our approach on three multi-satellite rendezvous problems, and Section \ref{sec:conclusion} provides concluding remarks.

\section{Preliminaries and Problem Formulation} \label{sec:prelim} 

\subsection{The Multivariate t Distribution} \label{ssec:t_dist}

The multivariate $t$ distribution is the vector generalization of the Student's $t$ distribution \cite{kotz_nadarajah_2004}.  The multivariate $t$ distribution encompasses a family of distributions characterized by parameters describing location, correlation structure, and how heavy tailed a distribution is. It is defined as follows. 
\begin{defn}[\cite{kotz_nadarajah_2004}]
A $n$-dimensional multivariate random variable that elicits the pdf
\begin{equation}
    \phi(\bvec{x}) = \frac {\Gamma \left((\nu + n)/2\right) \nu ^{\nu/2}}{\Gamma(\nu /2) \pi^{n/2} \mathrm{det}(\overline{\Sigma})^{\frac{1}{2}}}  \left[\nu+(\bvec{x}-\vec{\mu})^{T}\overline{\Sigma}^{-1}(\bvec{x} -\vec{\mu})\right]^{-\frac{\nu + n}{2}} 
\end{equation}
is said to have a multivariate $t$ distribution with location $\vec{\mu} \in \R^n$, scale matrix $\overline{\Sigma}\in \R_{++}^{n \times n}$, degrees of freedom $\nu\in \N^{+}$.
\end{defn}
The degree of freedom parameter, $\nu$, is a quantitative measure of how heavy the tails of the distribution are. Lower values of $\nu$ correspond to heavier tails. Note that $\nu = 1$ corresponds to the multivariate Cauchy distribution and the limiting distribution, as $\nu \rightarrow \infty$, is the multivariate Gaussian.

The common way in which a $n$ dimensional multivariate $t$ random variable is constructed is
\begin{equation}
    \bvec{x} \equiv \frac{\overline{\Sigma}^{\frac{1}{2}}\bvec{y}}{\sqrt{\boldsymbol{z}/\nu}} + \vec{\mu}
\end{equation}
where $\bvec{y} \sim Normal(\vec{0}, \overline{I}_n)$ and $\boldsymbol{z} \sim CSquare(\nu)$. Several properties can be derived from this construction. We outline the properties of marginal distributions and affine transformations in the Appendix to preface their use later.

It is clear from this construction that while the random vector's elements may by uncorrelated, they are not independent. In many practical applications, independence assumptions are used for mathematical convenience but may not accurately model the underlying circumstances. Consider a thrusters of a large aircraft. In the early stages of a flight the metal casing surrounding the thrusters will not have reached peak temperature. This time varying parameter effects the overall efficiency of the aircraft's engine \cite{Braun2018}. This time dependent structure, if incorporated into the model though the disturbance, makes the independence assumption invalid. Similarly, we can consider the weight of the propellant in high delta-v satellite maneuvers such as orbital inclination changes. Since the acceleration is inversely proportional to the mass of the craft, stochastic perturbations for the amount of propellant used to complete a maneuver can directly impact the amount of propellant needed to complete future maneuvers. Hence, the disturbance can be considered state dependent.  

\subsection{The Beta Prime Distribution} \label{ssec:beta_prime_dist}

A beta prime random variable that will be leveraged for collision avoidance is defined as follows. 
\begin{defn}[\cite{johnson1995continuous}]
A non-negative univariate variate random variable that elicits the pdf
\begin{equation}
    \phi(\bvec{x}) = \frac{x^{-\gamma - 1} (1+x)^{-\gamma-\delta}}{\beta (\gamma,\delta)}
\end{equation}
is said to have a beta prime distribution with shape parameters $\gamma \in \R_+$ and $\delta \in \R_+$. 
\end{defn}
The shape parameters define the polynomial shape of the pdf and the quantile. Of note, when $\gamma \leq 1$, the quantile is strictly convex as the pdf is monotonically decreasing. However, when $\gamma > 1$, the quantile is only convex in the region $p\in\left[ \Phi(\frac{\gamma-1}{\delta +1}), 1 \right]$.  Further, when $\gamma=1$ and/or $\delta=1$, the cdf has an analytic form. However, in many cases analytic expressions of the cdf do not guarantee analytic expressions of the quantile.

To construct a beta prime random variable, take the ratio of two gamma random variables with the same scale rate, $\bvec{x} \equiv \frac{\boldsymbol{y}}{\boldsymbol{z}}$ where $\boldsymbol{y} \sim Gamma(\gamma, \theta)$ and $\boldsymbol{z} \sim Gamma(\delta, \theta)$. The beta prime distribution, also known as the inverted beta or gamma ratio distribution, is commonly used in Bayesian analysis as the conjugate prior for a beta random variables and in financial modeling as a distribution to assess risk odds. In our proposed method, we use the beta prime distribution in our reformulation of the collision avoidance constraints. As we show in the Appendix, the beta prime distribution can arise by taking the squared 2-norm of a standard multivariate $t$ distributed random variable and dividing by the degrees of freedom, $\nu$. 

Upon reformulation of the intervehicle collision avoidance chance constraint, our function of interest will include the sum of two independent beta prime random variables. The infinite divisibility properties of the beta prime distribution will allow us to maintain a common form. We outline the infinite divisibility property in the Appendix.

\subsection{Problem Formulation} \label{ssec:problem_formulation}

We consider a discrete, linear, time-invariant system given by 
\begin{equation}
    \bvec{x}(k+1) = \overline{A} \bvec{x}(k) + \overline{B} \vec{u}(k) + \bvec{w}(k) \label{eq:system_one_step}
\end{equation}
with state $\bvec{x}(k) \in \mathscr{X} \subseteq \R^n$, input $\vec{u}(k) \in \mathscr{U} \subset \R^m$, and discrete time index $k \in \N_{N}$. Initial conditions, $\vec{x}(0)$, are assumed to be known and the set $\mathscr{U}$ is a convex polytope. 
\begin{assm} \label{assm:disturbance}
The disturbance, $\bvec{w}(k)$, is a multivariate $t$ distributed random vector,
\begin{equation}
        \bvec{w}(k)  \sim  t \left(\vec{\mu}, \overline{\Sigma}, \nu \right)
\end{equation}
for $\vec{\mu} \in \R^n$, $\overline{\Sigma} \in \R_{++}^{n \times n}$, and $\nu \in \N^{+}$.
\end{assm}
As the location parameter $\vec{\mu}$ is affine in the construction of the multivariate $t$ distribution, we will use $\vec{\mu}=\vec{0}$ to simplify derivations with no consequence. 

With a finite time horizon $N \in \N^{+}$, we can exploit the linearity of the system to rewrite the dynamics at time step $k$ as an affine summation of a transformed initial state, a concatenated input vector, and a concatenated disturbance vector:
\begin{equation} \label{eq:system_k_step}
    \bvec{x}(k) = \overline{A}^k \vec{x}(0) + \overline{\mathcal{C}}(k) \vec{U} + \overline{\mathcal{D}}(k) \bvec{W} \quad \forall k \in \N_{N}
\end{equation}
with 
\begin{subequations}
\begin{alignat}{2}
    \vec{U} =& \begin{bmatrix} \vec{u}(0)^\top & \ldots  & \vec{u}(N-1)^\top \end{bmatrix}^\top &&\in \mathscr{U}^{N} \\
    \bvec{W} =& \begin{bmatrix} \bvec{w}(0)^\top & \ldots & \bvec{w}(N-1)^\top \end{bmatrix}^\top &&\in \mathbb{R}^{Nn} \\
    \overline{\mathcal{C}}(k) = & \begin{bmatrix} \overline{A}^{k-1} \overline{B} &  \ldots &  \overline{A} \overline{B} &  \overline{B} &  \overline{0}_{n \times (N-k)m} \end{bmatrix} && \in \R^{n \times Nm} \\
    \overline{\mathcal{D}}(k) = & \begin{bmatrix} \overline{A}^{k-1} & \ldots & \overline{A} & \overline{I}_n & \overline{0}_{n \times (N-k)n} \end{bmatrix} && \in \R^{n \times Nn}
\end{alignat}
\end{subequations}   

Consider the following simplification. 
\begin{assm}\label{assm:distdepend}
The concatenated disturbance, $\bvec{W}$, is a multivariate $t$ distributed random vector,
\begin{equation} \label{eq:dist_simplified}
    \bvec{W} \sim  t \left(\vec{0}, \overline{\Psi}, \nu \right) 
\end{equation}
where 
\begin{equation}
    \overline{\Psi} = \mathrm{diag}\left(\underbrace{\overline{\Sigma}, \dots, \overline{\Sigma}}_{N \text{ times}} \right)
\end{equation} 
\end{assm}
This simplifying assumption is still in the spirit of Assumption \ref{assm:disturbance} as we can recover the conditions by the marginal properties of the multivariate $t$ distribution (see Property \ref{prop:t_margin} in Section \ref{ssec:t_dist}). As discussed in Section \ref{ssec:t_dist}, this implies the additive noise variables are not independent. The scenarios in which heavy tailed phenomena tend to appear are by nature atypical. For instance, faulty thrusters may lead to state dependent disturbances.  In scenarios like these, the non-independent assumption will likely be valid.

Consider the planning context in which $v$ vehicles evolve in bounded region with dynamics \eqref{eq:system_one_step} with state $\bvec{x}_i(k)$ and concatenated input $\vec{U}_i$ for vehicle $i$. We presume each vehicle has a potentially time-varying desired target set and must maintain some distance from other vehicles as well as static objects in the environment. Each restriction must hold with desired likelihood,
\begin{subequations}\label{eq:constraints}
    \begin{align}
    \Prob \{\bvec{x}_i(k)  \in  \mathscr{T}_{i,k} \} & \geq  1-\alpha_{\mathscr{T}} \label{eq:constraint_t}\\
    \Prob\{\boldsymbol \| \overline{S} (\bvec{x}_i(k) -  \vec{o}) \| \geq r\} &\geq 1-\alpha_o, \: \forall i \in \N^{+}_{v} \label{eq:constraint_o}\\
    \Prob\{\boldsymbol \| \overline{S} (\bvec{x}_i(k) - \bvec{x}_j (k) )\| \geq r\} &\geq 1\!-\!\alpha_r, \: \forall i \!\neq\! j \in \N^{+}_{v} \label{eq:constraint_r}
    \end{align}
\end{subequations}
where  $ \mathscr{T}_{i,k} \subseteq \R^n$ are convex, compact, and polytopic sets, $\overline{S} = \begin{bmatrix} \overline{I}_q & \overline{0}_{q \times n-q} \end{bmatrix} $ is a matrix designed to extract positional elements from the state, $r \in \R_+$, $\overline{o} \in \R^n$ are static object locations, and $\alpha_{\mathscr{T}} \in (0,0.5]$,  $\alpha_{o}, \alpha_r, \in (0,1)$ are desired probabilistic violation thresholds.
\begin{assm} \label{assm:convex}
The probabilistic violation thresholds, $\alpha_{o}$ and $\alpha_r$, are small enough to maintain convexity of the proposed reformulation of \eqref{eq:constraints}.
\end{assm}
As discussed in Section \ref{ssec:beta_prime_dist} convexity of the problem will depend both on $\nu$ and $q$. If $\gamma \leq 1$, any value will suffice. However, establishing convexity can be challenging when $\gamma > 1$.

We seek to minimize a convex performance objective $J: \mathscr{X}^{N \times v} \times \mathscr{U}^{N \times v} \rightarrow \R$. 
\begin{subequations}\label{prob:big_prob_eq}
    \begin{align}
        \underset{\vec{U}_1, \dots, \vec{U}_v}{\mathrm{minimize}} \quad & J\left(
        \bvec{X}_1, \ldots, \bvec{X}_v,  \vec{U}_1, \dots, \vec{U}_v\right)  \\
        \mathrm{subject\ to} \quad  & \vec{U}_1, \dots, \vec{U}_v \in  \mathscr{U}^N,  \\
        & \text{Dynamics } \eqref{eq:system_k_step} \text{ with }
        \vec{x}_1(0), \dots, \vec{x}_v(0)\\
        & \text{Probabilistic constraints  \eqref{eq:constraints}} \label{eq:big_prob_eq_d}
    \end{align}
\end{subequations}
where $\bvec{X}_i = \begin{bmatrix} \bvec{x}_i^{\top}(1) & \ldots & \bvec{x}_i^{\top}(k) \end{bmatrix}^{\top}$ is the concatenated state vector for vehicle $i$.

Here, we address the following problem.
\begin{prob} \label{prob:1}
Solve the stochastic motion planning problem \eqref{prob:big_prob_eq} with open loop controllers $\vec{U}_1, \dots, \vec{U}_v \in \mathscr{U}^N$, for predetermined probabilistic violation thresholds $\alpha_{\mathscr{T}}, \alpha_o, \alpha_r$.
\end{prob}
We do so by solving two sub-problems.
\setcounter{prob}{1}
\begin{subprob} \label{prob:1.1}
Reformulate \eqref{eq:constraints} into a form that guarantees satisfaction and allows for convex optimization techniques.
\end{subprob}
\begin{subprob}\label{prob:1.2}
Determine minimum values for $\alpha_o$ and $\alpha_r$ that ensure convexity of Problem \ref{prob:1.1}.
\end{subprob}
The main challenge in solving Problem \ref{prob:1} is solving Problem \ref{prob:1.1}.

\section{Methods}\label{sec:methods}

We solve Problem \ref{prob:1} with standard risk allocation techniques \cite{ono2008iterative} in conjunction with quantile reformulations. We generate a piecewise affine approximation of the quantile via a Taylor series approximation in the convex region  of the quantile. We embed the piecewise affine approximation in a difference-of-convex functions framework to iteratively solve the reverse convex constraints to a local minimum. The difference of convex functions framework enable efficient optimization by quadratic programs. 

\begin{defn}[Reverse convex constraint]
A reverse convex constraint is the complement of a convex constraint, that is, $f(x) \geq c$ for a convex function $f: \R \rightarrow \R$ and a scalar $c \in \mathbb R$.
\end{defn}

\subsection{Reformulation of constraints} \label{ssub:reform}

We start by noting that each constraint in \eqref{eq:constraints} can be rewritten in the form,
\begin{subequations}\label{eq:prob_constraints}
\begin{align}
    \Prob \left\{ \bigcap_{i=1}^{n_c} f_i(\vec{x}(0), \vec{U}) + g_i \boldsymbol{y}_{i} \leq c_i \right\} & \geq 1-\alpha \label{eq:prob_constraints_convex} \\
    \Prob \left\{ \bigcap_{i=1}^{n_c} f_i(\vec{x}(0), \vec{U}) - g_i \boldsymbol{y}_{i} \geq c_i \right\} & \geq 1-\alpha \label{eq:prob_constraints_reverse_convex}
\end{align}
\end{subequations}
where $f: \mathscr{X} \times \mathscr{U}^N \rightarrow \R$ is convex in $\vec{U}$, $g_i \in \R_+$ is a positive scalar, $\boldsymbol{y}_{i}$ is a real and continuous random variable that is a function of the disturbance, and $n_c$ is the number of constraints that must jointly be satisfied. We presume $c_i$ is a constant and $\alpha$ is a predetermined probabilistic violation threshold. This form will allow us to allocate risk for each individual chance constraint and facilitate the reformulation into the a solvable form. We outline the reformulation of \eqref{eq:constraints} into the form of \eqref{eq:prob_constraints}.

\subsubsection{Target Constraint Reformulation}

Consider a target set constraint in the form \eqref{eq:constraint_t}. The polytopic construction implies there exists some matrix, $\overline{P} \in \R^{L \times n}$, and vector, $\vec{q} \in \R^L$ such that
\begin{equation}
    \bvec{x}(k) \in \mathscr{T}_{k} \equiv \overline{P} \bvec{x}(k) \leq \vec{q}_i.
\end{equation}
Using half-space form we can reformulate \eqref{eq:constraint_t} into a probabilistic constraint in the form of \eqref{eq:prob_constraints_convex}. We reformulate the target set constraint as 

\begin{subequations} \label{eq:target_reform}
\begin{align}
    \Prob \{\bvec{x}(k)  \in  \mathscr{T}_{k} \} 
    = & \;  \Prob \{\overline{P} \bvec{x}(k)  \leq \vec{q} \}  \\
    = & \;  \Prob \left\{ \bigcap_{i=1}^{L} \vec{P}_{i} \bvec{x}(k)  \leq q_{i} \right\}  \\
    = & \;  \Prob \left\{ \bigcap_{i=1}^{L} \vec{P}_{i} \left(\overline{A}^k \vec{x}(0) + \overline{\mathcal{C}}(k) \vec{U}_i + \overline{\mathcal{D}}(k) \bvec{W}_i\right)  \leq q_{i} \right\} \\
    = & \;  \Prob \Bigg\{ \bigcap_{i=1}^{L} \underbrace{\vec{P}_{i} \left(\overline{A}^k \vec{x}(0) + \overline{\mathcal{C}}(k) \vec{U} \right)}_{f_i(\vec{x}(0), \vec{U})} + \underbrace{\left(\vec{P}_{i} \overline{\mathcal{D}}(k) \overline{\Psi} \overline{\mathcal{D}}(k)^{\top} \vec{P}_{i}^{\top} \right)^{\frac{1}{2}}}_{g_i} \underbrace{\tau}_{\boldsymbol{y}_{i}} \leq \underbrace{q_{i}}_{c_i} \Bigg\}  \label{eq:target_reform_e}
\end{align}
\end{subequations}

where  $\tau \sim t(0, 1, \nu)$. Thus,
\begin{equation} \label{eq:target_equiv}
    \Prob \{\bvec{x}(k)  \in  \mathscr{T}_{k} \} \geq 1 - \alpha \Leftrightarrow \eqref{eq:target_reform_e} \geq 1 - \alpha
\end{equation}
where the random variable $\tau$ has a univariate Student's $t$ distribution.

\subsubsection{Collision Avoidance Constraint Reformulation}

We consider the probabilistic collision avoidance between two vehicles, $i$ and $j$, respectively as being the 2-norm distance being greater than or equal to a predetermined distance, $r$, with probability of at least $1-\alpha$,
\begin{equation} \label{eq:ref_avoid}
    \Prob \{ \| \overline{S} (\bvec{x}_i(k) - \bvec{x}_j(k))\| \geq r \} \geq 1 - \alpha
\end{equation}

We reformulate \eqref{eq:ref_avoid} as
\begin{subequations}\label{eq:collision_reform}
\begin{align}
    & \; \Prob \left\{    \left\| \overline{S} \left(\bvec{x}_i(k) - \bvec{x}_j(k)\right)\right\| \geq r \right\} \label{eq:collision_reform_a} \\
    = & \; \Prob \left\{   \left\| \overline{S} \left(\overline{A}^k \vec{x}_{i- j}(0)  +  \overline{\mathcal{C}}(k) \vec{U}_{i- j} + \overline{\mathcal{D}}(k) \bvec{W}_{i} - \overline{\mathcal{D}}(k) \bvec{W}_{j} \right)\right\| \geq r\right\}  \\
    = & \; \Prob \left\{   \left\| \overline{S} \left(\overline{A}^k \vec{x}_{i- j}(0)  +  \overline{\mathcal{C}}(k) \vec{U}_{i- j}\right)  + \left( \overline{S} \overline{\mathcal{D}}(k)\Psi \overline{\mathcal{D}}(k)^{\top} \overline{S}^{\top} \right)^{\frac{1}{2}} \left(\bvec{\tau}_{i}  -  \bvec{\tau}_{j}\right) \right\| \geq r\right\}  \\
    \geq & \; \Prob \left\{   \left\| \overline{S} \left(\overline{A}^k \vec{x}_{i- j}(0)  +  \overline{\mathcal{C}}(k) \vec{U}_{i- j}\right)\right\|  -  \left\| \left( \overline{S} \overline{\mathcal{D}}(k)\Psi \overline{\mathcal{D}}(k)^{\top} \overline{S}^{\top} \right)^{\frac{1}{2}} \left(\bvec{\tau}_{i}  -  \bvec{\tau}_{j}\right) \right\| \geq r \right\}    \label{eq:collision_reform_d}\\
    \geq & \; \Prob \left\{   \left\| \overline{S} \left(\overline{A}^k \vec{x}_{i- j}(0)  +  \overline{\mathcal{C}}(k) \vec{U}_{i- j}\right)\right\|  - \sqrt{\lambda_{\max}\left( \overline{S} \overline{\mathcal{D}}(k)\Psi \overline{\mathcal{D}}(k)^{\top} \overline{S}^{\top} \right)}\left\|   \left(\bvec{\tau}_{i}  -  \bvec{\tau}_{j}\right) \right\| \geq  r \right\}  \label{eq:collision_reform_e}\\
    \geq & \; \Prob \left\{   \left\| \overline{S} \left(\overline{A}^k \vec{x}_{i- j}(0)  +  \overline{\mathcal{C}}(k) \vec{U}_{i- j}\right)\right\|  - \sqrt{2\lambda_{\max}\left( \overline{S} \overline{\mathcal{D}}(k)\Psi \overline{\mathcal{D}}(k)^{\top} \overline{S}^{\top} \right)} \sqrt{\left\|\bvec{\tau}_{i}\right\|^2 + \left\| \bvec{\tau}_j \right\|^2} \geq  r \right\}  \label{eq:collision_reform_f}\\
    = & \; \Prob \left\{  \underbrace{ \left\| \overline{S} \left(\overline{A}^k \vec{x}_{i- j}(0)  +  \overline{\mathcal{C}}(k) \vec{U}_{i- j}\right)\right\|}_{f_i(\vec{x}(0), \vec{U})}  - \underbrace{\sqrt{2 \nu \lambda_{\max}\left( \overline{S} \overline{\mathcal{D}}(k)\Psi \overline{\mathcal{D}}(k)^{\top} \overline{S}^{\top} \right)}}_{g_i} \underbrace{ \sqrt{\frac{1}{\nu}\left(\left\|\bvec{\tau}_{i}\right\|^2 + \left\| \bvec{\tau}_{j} \right\|^2 \right)}}_{\boldsymbol{y}_{i}} \geq  \underbrace{r}_{c_i} \right\} \label{eq:collision_reform_g}
\end{align}
\end{subequations}
where $\boldsymbol{\tau}_i \sim t(\overline{0}, I_{q}, \nu)$. Here \eqref{eq:collision_reform_d} employs the reverse triangle inequality, \eqref{eq:collision_reform_e} employs the variational properties of matrices, and \eqref{eq:collision_reform_f} employs the parallelogram law. Satisfaction of \eqref{eq:collision_reform_g} implies satisfaction of \eqref{eq:collision_reform_a}. Here, 
\begin{subequations} \label{eq:Q_dist}
\begin{align}
    \boldsymbol{y}_{i}^2 \sim & BPrime (\gamma, \delta)\\
    \gamma = & \frac { 2q \left(\frac{q}{2} +\left(\frac{\nu}{2}\right) ^{2}-\nu +\left(\frac{q\nu}{2}\right)-2q+1 \right)}{\left(\nu-2\right)\left(q+\nu-2\right)}\\
    \delta = & \frac {2\left(-q +\left(\frac{\nu}{2}\right) ^{2}-\left(\frac{\nu}{2}\right) + \left(\frac{q\nu}{2}\right) \right)}{q+\nu-2} 
\end{align}
\end{subequations}
by Properties \ref{prop:beta_prime_norm} and \ref{prop:beta_prime_sum} in the Appendix. We recover the pdf of $\boldsymbol{y}_{i}$ as 
\begin{equation}
    \phi_{\boldsymbol{y}_{i}}(x) = 2x\phi_{\boldsymbol{y}_{i}^2}(x^2)
\end{equation}
to be used in the quantile approximation in Section \ref{ssec:quantile_approx}.

We note that the reformulation of \eqref{eq:constraint_o} follows similar to \eqref{eq:collision_reform}. Since the object is static in the environment, we skip \eqref{eq:collision_reform_f} to get the probabilistic constraint into the form of a known distribution. The resulting distribution is $BPrime \left(\frac{q}{2}, \frac{\nu}{2}\right)$. 

Hence, we can conservatively approximate solutions to Problem \ref{prob:1} by rewriting it in terms of \eqref{eq:prob_constraints}.
\begin{prob} \label{prob:2}
Solve the optimization problem
\begin{subequations}
    \begin{align}
        \underset{\vec{U}_1, \dots, \vec{U}_v}{\mathrm{minimize}} \quad & J\left(
        \bvec{X}_1, \ldots, \bvec{X}_v,  \vec{U}_1, \dots, \vec{U}_v\right)  \\
        \mathrm{subject\ to} \quad  & \vec{U}_1, \dots, \vec{U}_v \in  \mathscr{U}^N,  \\
        & \text{Dynamics } \eqref{eq:system_k_step} \text{ with }
        \vec{x}_1(0), \dots, \vec{x}_v(0)\\
        & \text{Probabilistic constraints  \eqref{eq:prob_constraints}} \label{eq:prob2_final}
    \end{align}
\end{subequations}
with open loop controllers $\vec{U}_1, \dots, \vec{U}_v \in \mathscr{U}^N$, for predetermined probabilistic violation thresholds.
\end{prob}

This new optimization problem differs from \eqref{prob:big_prob_eq} in the formulation of the chance constraints. While this is a small change, it facilitates a reformulation that eliminates the need to evaluate high dimensional integrals for computing the probabilities of constraint satisfaction.

\begin{lem} \label{lem:1}
Any solution to Problem \ref{prob:2} is a sub-optimal solution to Problem \ref{prob:1}.
\end{lem}

\begin{proof}
By setting the probabilistic threshold in \eqref{eq:target_equiv} to $\alpha_\mathscr{T}$ satisfaction of \eqref{eq:target_reform_e} is equivalent to \eqref{eq:constraint_t}. The reformulation \eqref{eq:collision_reform} is a tightening of the constraints \eqref{eq:constraint_o}-\eqref{eq:constraint_r} when we set $\alpha$ in \eqref{eq:ref_avoid} to $\alpha_o$ and $\alpha_r$, respectively. Therefore, satisfaction of \eqref{eq:prob2_final} implies satisfaction of \eqref{eq:big_prob_eq_d} 
\end{proof}

Lemma \ref{lem:1} dictates that \eqref{eq:prob_constraints} are conservative bounds for \eqref{eq:constraints}. How conservative each bound is will vary based on the bounding mechanisms used. For example, the target set constraint reformulation \eqref{eq:target_reform} created a tight bound as we are simply using the affine properties of the distribution. However, the collision avoidance constraint reformulation \eqref{eq:collision_reform} introduces significant conservatism through the the reverse triangle inequality and the parallelogram law. Empirically, we have observed solutions being 10-20\% more conservative when comparing violation thresholds to Monte Carlo satisfaction rates \cite{PrioreACC21, prioreACC22}.

\subsection{Constraint satisfaction via quantiles}

To solve Problem \ref{prob:2}, we employ the quantile reformulation with a standard risk allocation framework via Boole's inequality \cite{casella2002}. Here, we derive the quantile form for \eqref{eq:prob_constraints_reverse_convex} and note that the reformulation for \eqref{eq:prob_constraints_convex} is analogous. We start by taking the complement of \eqref{eq:prob_constraints_reverse_convex} via the complementary properties of probabilities and De Morgan's law \cite{casella2002},
\begin{equation}
    \Prob \left\{ \bigcup_{i=1}^{n_c} f_i(\vec{x}(0), \vec{U}) - g_i \boldsymbol{y}_{i} \leq c_i \right\} \leq \alpha
\end{equation}
We upper bound this probability via Boole's inequality,
\begin{equation}
    \Prob \left\{\bigcup_{i=1}^{n_c} f_i(\vec{x}(0), \vec{U}) - g_i \boldsymbol{y}_{i} \leq c_i \right\} \leq \sum_{i=1}^{n_c} \Prob \left\{f_i(\vec{x}(0), \vec{U}) - g_i \boldsymbol{y}_{i} \leq c_i \right\} \label{eq:boole_right}
\end{equation}
Using the approach in \cite{ono2008iterative}, we introduce variables $\eta_{i}$ to allocate risk to each of the individual probabilities
\begin{subequations}\label{eq:quantile_reform_new_var}
\begin{align}
     \Prob \left\{ f_i(\vec{x}(0), \vec{U}) - g_i \boldsymbol{y}_{i} \leq c_i  \right\} &\leq \eta_{i} \label{eq:quantile_orig} \\
     \sum_{i=1}^{n_c} \eta_{i} &\leq \alpha \label{eq:quantile_reform_new_var_2}\\
     \eta_{i} & \geq 0 \label{eq:quantile_reform_new_var_3}
\end{align}
\end{subequations}
By inverting the argument of \eqref{eq:quantile_orig}, we obtain
\begin{equation}
    \Prob \left\{ \boldsymbol{y}_{i} \leq \frac{1}{g_i}\left(f_i(\vec{x}(0), \vec{U}) -c_i\right) \right\} \geq 1-\eta_{i} 
    \Leftrightarrow \frac{1}{g_i}\left(f_i(\vec{x}(0), \vec{U})-c_i \right)  \geq \Phi^{-1}_{\boldsymbol{y}_{i}}\left(1-\eta_{i} \right)
\label{eq:quantile_opt}
\end{equation}
Rearranging \eqref{eq:quantile_opt}, we obtain
\begin{equation} \label{eq:reformed_rev}
     f_i(\vec{x}(0), \vec{U}) - g_i\left(\Phi^{-1}_{\boldsymbol{y}_{i}}\left(1-\eta_{i} \right)\right)  \geq c_i 
\end{equation}
which is a reverse convex constraint in $\vec{U}$.

We note that reformulation of \eqref{eq:prob_constraints_convex} deviates from \eqref{eq:reformed_rev} only by making the signs of $f_i(\vec{x}(0), \vec{U})$ and $c_i$ negative. Hence, we get the series of inequalities
\begin{subequations}\label{eq:quantile_final}  
\begin{align} 
     f_i(\vec{x}(0), \vec{U})  & \leq c_i - g_i \left(\Phi^{-1}_{\boldsymbol{y}_{i}}\left(1-\upsilon_{i} \right)\right) \label{eq:quantile_reform_final_1}\\
     \sum_{i=1}^{n_c} \upsilon_{i} &\leq \alpha  \label{eq:quantile_reform_final_2}\\
     \upsilon_{i} & \geq 0 \label{eq:quantile_reform_final_3}\\ 
     f_i(\vec{x}(0), \vec{U})   &\geq c_i +g_i \left(\Phi^{-1}_{\boldsymbol{y}_{i}}\left(1-\eta_{i} \right)\right) \label{eq:quantile_reform_final_4}\\
     \sum_{i=1}^{n_c} \eta_{i} &\leq \alpha \label{eq:quantile_reform_final_5}\\
     \eta_{i} & \geq 0\label{eq:quantile_reform_final_6}
\end{align}
\end{subequations}
that can be substituted into Problem \ref{prob:2} in place of \eqref{eq:prob2_final}.

\begin{lem} \label{lem:2}
For the controller $\vec{U}_1, \dots, \vec{U}_v$, if there exists risk allocation variables $\upsilon_{i}$ satisfying \eqref{eq:quantile_reform_final_2}-\eqref{eq:quantile_reform_final_3} for constraints in the form of \eqref{eq:quantile_reform_final_1} and risk allocation variables $\eta_{i}$ satisfying \eqref{eq:quantile_reform_final_5}-\eqref{eq:quantile_reform_final_6} for constraints in the form of \eqref{eq:quantile_reform_final_4}, then $\vec{U}_1, \dots, \vec{U}_v$ satisfy \eqref{eq:prob2_final}.
\end{lem}

\begin{proof}
Satisfaction of \eqref{eq:quantile_reform_final_2}-\eqref{eq:quantile_reform_final_3} and \eqref{eq:quantile_reform_final_5}-\eqref{eq:quantile_reform_final_6} implies \eqref{eq:boole_right} meets the probabilistic violation threshold of $\alpha$. Boole's inequality and De Morgan's laws guarantee \eqref{eq:prob2_final} is satisfied.
\end{proof}

The constraint \eqref{eq:quantile_reform_final_1} is convex in $\vec{U}$ and $\upsilon_{i}$, however \eqref{eq:quantile_reform_final_4} is reverse convex in $\vec{U}$ and convex in $\eta_{i}$.  
Additionally, while Assumption \ref{assm:convex} guarantees the convexity of \eqref{eq:quantile_reform_final_1} and \eqref{eq:quantile_reform_final_4} in $\upsilon_{i}$ and $\eta_{i}$, respectively, the expressions $\Phi^{-1}_{\boldsymbol{y}_{i}}(1-\upsilon_{i})$ and $\Phi^{-1}_{\boldsymbol{y}_{i}}(1-\eta_{i})$ are non-conic.

\subsection{Numerical quantile approximation} \label{ssec:quantile_approx}

To make the expressions $\Phi^{-1}_{\boldsymbol{y}_{i}}(1-\upsilon_{i})$ and $\Phi^{-1}_{\boldsymbol{y}_{i}}(1-\eta_{i})$ amenable to standard convex and conic techniques we generate a piece-wise affine approximation such that
\begin{equation}
    \Phi^{-1}_{\boldsymbol{y}_{i}}(1-\eta_{i}) \approx \max_{q \in \N_{l^{\ast}}^{+}} (m_{iq} \, \eta_{i} + c_{iq})
\end{equation}
where $m_{iq}$ is the piece-wise slope, $c_{iq}$ is the piece-wise intercept, and $l^\ast$ is the cardinality of the set $\mathscr{Q} = \{\{m_{i1},  c_{i1}\}, \ldots, \{m_{il^{\ast}},  c_{il^{\ast}}\} \}$. This piece-wise affine form allows up to select the convex position of the quantile and enbed it in the optimization problem in place of the potentially non-conic quantile expression. To generate the set $\mathscr{Q}$ we must first compute a series of points to linearize between. The quantile for the Student's $t$ distribution used in the target set constraint does not have an analytic expression as it requires the inversion of the Gauss hypergeometric function. Evaluating the quantile of the beta prime distribution used for collision avoidance constraints will require the inversion the incomplete beta function. In many cases, the incomplete beta function will not have an analytic form making the inversion impossible. The standard method for computing quantiles for the $t$ distribution is through nested summation over a series of cosines and has differing implementations for even or odd $\nu$ \cite{Hill1970}. The quantile of the beta prime distribution is based on computation of the beta quantile with a modified Newton-Raphson iterative method \cite{Cran1977}. Both approximations are highly tailored to their respective distributions but lack the generality needed to handle potentially varying parameterizations. To generate the series of points required to linearize each distributions, we opt for an alternative method that is amenable to both distributions and can accommodate potentially varying parameterizations. The approach relies on a Taylor series approximation of the quantile \cite{Yu2017}.

For a random variable $\boldsymbol{y}_i$, and a known quantile evaluation point $\boldsymbol{y}_i^{\ast} = \Phi^{-1}_{\boldsymbol{y}_i}(p_{0})$ for $p_0 \in (0,1)$, \cite{Yu2017} proposes an iterative process that evaluates a Taylor series expansion of $n_d$ terms at points that are an interval $h \in \mathbb R_+$ apart. A quantile approximation at $p_{c+1} = p_c + h$ is described by
\begin{equation}\label{eq:taylor_quantile}
    \hat{\Phi}^{-1}_{\boldsymbol{y}_i}(p_{c+1}) =  \Phi^{-1}_{\boldsymbol{y}_i}(p_c) + \sum_{d=1}^{n_d} (-1)^{d}\left. \frac{\partial^d \Phi^{-1}_{\boldsymbol{y}_i}(p)}{(\partial \kappa)^d} \right|_{p = p_{c}} \cdot \frac{\log(p_{c+1}/p_c)^d}{d!}
\end{equation}  
where $\kappa =- \log(p)$ is a variable substitution used for numerical tractability.  Typically, $n_d=3$ or 4 is sufficient to generate an approximation with small errors, and steps $c$ are computed until a terminating percentile is reached. Derivatives of the quantile are obtained via the inverse function theorem,
\begin{equation}
    \frac{\partial}{\partial \kappa} \Phi^{-1}_{\boldsymbol{y}_i}(p) = -\frac{e^{-\kappa}}{\phi_{\boldsymbol{y}_i} \left(\Phi^{-1}_{\boldsymbol{y}_i}(p) \right)}
\end{equation}
where the $i$\textsuperscript{th} derivative will elicit the the $i-1$\textsuperscript{th} derivative of $\phi_{\boldsymbol{y}_i} (\cdot)$. Analytical expressions for the first four derivatives are provided in \cite{Yu2017}.

As the Student's $t$ distribution is symmetric about 0, we can easily define the instantiation point as $\boldsymbol{y}_{i}^{\ast}=0$ for $p_0 = 0.5$. As the beta prime has a skewed distribution, establishing an instantiation point becomes more involved. For instances where $\gamma=1$ or $\delta=1$, we can derive analytical expressions for the median as
\begin{subequations}
\begin{align}
    I_{\frac{0.5}{1+0.5}}(1, \delta) = & \; 2^{\frac{1}{\delta}} - 1 \label{eq:beta_prime_med_a}\\
    I_{\frac{0.5}{1+0.5}}(\gamma, 1) = & \; \left(2^{\frac{1}{\gamma}} - 1\right)^{-1} 
\end{align}
\end{subequations}
where $I_{\frac{0.5}{1+0.5}}(\gamma, \delta)$ is the median of the beta prime distribution. Similarly, when $\gamma=\delta$, the median is $1$.  For all other cases, we will need to approximate the median. By construction, a beta prime random variable can be constructed as the quotient of two independent gamma random variables. We can approximate the median by looking at the ratio of the medians of the gamma random variables \cite{kerman2011}. Several approximations of the median of a gamma random have been proposed for the case when $\gamma \neq \delta \neq 1$ \cite{Berg2006, GAUNT2021, Lyon2021}. Most notable approximations have an maximum absolute error near 2\%. The optimal choice of approximation will depend highly on the relationship between $\gamma$ and $\delta$. Note, the variable of interest is the square root of a beta prime distributed random variable. It can be shown that the median of this new random variable is the square root of the median of the beta prime random variable.

For an accurate instantiation point, error in the approximation 
\begin{equation}\label{eq:taylor_error_eq}
    \epsilon = \Phi^{-1}_{\boldsymbol{y}_{i}}(\cdot) - \hat{\Phi}^{-1}_{\boldsymbol{y}_{i}}(\cdot) 
\end{equation}  
is characterized by the unused Taylor series terms, such that
\begin{equation} \label{eq:taylor_error}
    \epsilon \in O\left([h / \min (p_{l-1}, p_0)]^{n_d}\right)
\end{equation} 
so that $\epsilon$ converges to $0$ as $h \rightarrow 0$ and $n_d \rightarrow \infty$ \cite{Yu2017}.

We compute the set $\{m_{iq},  c_{iq}\}$ to connect evaluation points of the quantile approximation.  However, to ensure a reasonable number of variables and constraints in the optimization, we selectively choose evaluation points, rather than connecting all points.  Given an error threshold, $\xi$, we seek a subset of $l^\ast$ affine terms, such that 
\begin{equation} \label{eq:quantile_over_error}
    \hat{\Phi}^{-1}_{\boldsymbol{y}_i}(p_c) \leq \max_{q \in \N_{l^{\ast}}^{+}} (m_{iq} \, \eta_{i} + c_{iq}) \leq \hat{\Phi}^{-1}_{\boldsymbol{y}_i}(p_c) + \xi 
\end{equation}
for all $p_c$, as shown in Figure \ref{fig:approx_graph}. We propose Algorithm \ref{algo:PWA} to compute the reduced set $\{m_{iq}, c_{iq} \}$.  Note that although the error threshold, $\xi$, is formulated with respect to the approximation (not the true quantile), as $\epsilon \rightarrow 0$, the convexity of the Student's $t$ quantile in the range $[0.5,1)$ guarantees that \eqref{eq:quantile_over_error} becomes an affine overapproximation of the true quantile. Similarly, Assumption \ref{assm:convex} guarantees that \eqref{eq:quantile_over_error} becomes an affine overapproximation for an accurate instantiation point for the beta prime distribution. 

To ensure Assumption \ref{assm:convex}, we can look at the relationship between the median approximation and the mode for instances where $\gamma > 1$. If the median is greater than the mode, $\alpha_o$ and $\alpha_r$ can take on any value greater than 0.5. If however, the median if less than the mode, we can only use our approximation in the region where the computed values are greater than the mode. Since the values are computed iteratively, checking violation thresholds against the lower bound on the convex region of the approximation is straight forward. 

\begin{figure}
\centering
\includegraphics[width=0.4\linewidth]{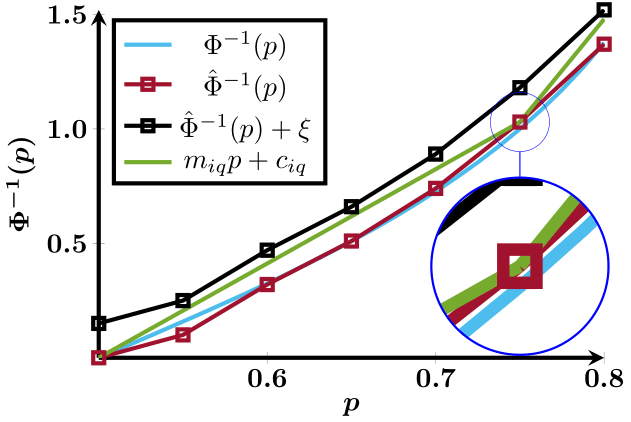}
\caption{Quantile approximation method applied to a Cauchy distribution.  The blue line represents the true quantile, the red points result from a Taylor series approximation (\ref{eq:taylor_quantile}), the black points show the error threshold $\xi$, and the green lines represent the affine approximation (\ref{eq:quantile_over_error}). 
}
\label{fig:approx_graph}
\end{figure}

\begin{algorithm}
  \caption{Computing $\{m_{iq}, c_{iq}\}$ from $\phi_{\boldsymbol{y}_i}$}
	\label{algo:PWA}
	\textbf{Input}: The pdf of $\boldsymbol{y}_i$, $\phi_{\boldsymbol{y}_i}$, and its derivatives $\phi^{'}_{\boldsymbol{y}_i}, \ldots, \phi^{(n)}_{\boldsymbol{y}_i}$, instantiating point $p_0$, termination point $p_l$, known quantile $\boldsymbol{y}_i^{\ast} = \Phi^{-1}_{\boldsymbol{y}_i}(p_0)$, step size $h$, and maximum error threshold $\xi$.
	\\
	\textbf{Output}: Affine terms of $\hat\Phi^{-1}_{\boldsymbol{y}_i}$, $\{m_{ijq}, c_{ijq}\}$
	\begin{algorithmic}[1]
	\For{$p_i = p_0+h$ \textbf{to} $p_l$ \textbf{by} $h$}
	\State $\mathcal P_i \gets \hat{\Phi}^{-1}(p_i)$ \Comment{Via \eqref{eq:taylor_quantile}}
	\EndFor
	\State	$i \gets 0$
	\For{$j = l$ \textbf{to} $i+1$ \textbf{by} $-1$}
	\State $m \gets \frac{\mathcal P_j-\mathcal P_i}{h(j-i)}$
	\State $c \gets \mathcal P_i - p_i \times m$
	\For{$y=i+1$ \textbf{to} $l-1$ \textbf{by} $1$}
	\State $\epsilon_y = \mathcal P_y - (p_y \times m + c)$
	\If{$\epsilon_y > \xi$}
	\State \textbf{Next} $j$
	\EndIf
	\EndFor
	\State $\{m_{iq}, c_{iq}\} \gets m, c$
	\State $i \gets j$
	\State Go to step $6$
	\EndFor
  \end{algorithmic}
\end{algorithm}

We reformulate \eqref{eq:quantile_reform_final_1} with the piecewise affine approximation \eqref{eq:quantile_over_error}, as
\begin{subequations}\label{eq:Num_quantile_reform}
\begin{align}
     f_i(\vec{x}(0), \vec{U})  &\leq c_i - \frac{1}{g_i}\left(s_{i}\right) & \label{eq:Num_quantile_reform_1}\\
     s_{i} & \geq m_{iq} \, \eta_{i} + c_{iq} & \forall q \in \N_{l^\ast}^{+}\\
     \sum_{j=1}^{v}\sum_{i=1}^{n_j} \eta_{i} &\leq \alpha & \\
     \eta_{i} & \geq 0 &
\end{align}
\end{subequations}
with slack variables $s_{ij}$.  A similar reformulation can be posed for \eqref{eq:quantile_reform_final_4}.
In the limit, as \eqref{eq:quantile_over_error} becomes an affine overapproximation of $\Phi^{-1}_{\boldsymbol{y}_i}(\cdot)$, \eqref{eq:Num_quantile_reform} is a tightening of \eqref{eq:quantile_final} and Assumption \ref{assm:convex} ensures the convexity of \eqref{eq:Num_quantile_reform}. 

\begin{lem} \label{lem:3}
Assume the instantiation point of \eqref{eq:taylor_quantile} is exact, i.e. $\Phi^{-1}(0.5)=0$ for both the standard Gaussian and Student's $t$ distribution. For a controller $\vec{U}_1, \dots, \vec{U}_v$, if there exists risk allocation variables $\eta_{i}$ and $\upsilon_{i}$, and slack variables $s_{i}$ satisfying \eqref{eq:Num_quantile_reform}, then $\vec{U}_1, \dots, \vec{U}_v$ asymptotically satisfies \eqref{eq:prob2_final} as $h \rightarrow 0$ and $n_d \rightarrow \infty$.
\end{lem}

\begin{proof} 
By \eqref{eq:taylor_error}, the approximation error $\epsilon \rightarrow 0$ as $h \rightarrow 0$ and $n_d \rightarrow \infty$. In this case, \eqref{eq:Num_quantile_reform} conservatively enforces \eqref{eq:quantile_final} by \eqref{eq:quantile_over_error}. By Lemma \ref{lem:2}, \eqref{eq:prob2_final} is conservatively enforced. 
\end{proof}

We note that a limitation of our approach is that we can only guarantee constraint satisfaction in the limit. In practice, a sufficiently differentiable distribution will likely behave well enough that four or more derivatives will result in an approximation with small errors given a small enough step size. We have found empirically that a step size, $h$, on the order of $10^{-6}$, is sufficiently small that the approximation error, \eqref{eq:taylor_error_eq}, is also on the order of $10^{-6}$ when comparing quantile approximation with closed form quantile functions. 

\subsection{Reverse convex constraints}

A standard approach to handling reverse convex constraints is difference of convex functions framework, 
\begin{equation} \label{eq:dc}
\begin{split}
    \underset{x}{\mathrm{minimize}} \quad & \mathcal{F}_0(x)-\mathcal{G}_0(x)  \\
    \mathrm{subject\ to} \quad  & \mathcal{F}_i(x)-\mathcal{G}_i(x) \leq 0 \quad \text{for } i \in \mathbb{N}_{[1,L]}  \\
\end{split}   
\end{equation}
in which the cost and constraints are represented as the difference of two convex functions, i.e., $\mathcal{F}_0, \mathcal{F}_i(\cdot): \R^n \rightarrow \R$ and $\mathcal{G}_0, \mathcal{G}_i(\cdot): \R^n \rightarrow \R$ for $x \in \R^n$ are convex. The convex-concave procedure solves \eqref{eq:dc} to a local minimum \cite{boyd_dc_2016} through an iterative approach, which employs first order approximations of $\mathcal{G}_0, \mathcal{G}_i$ at each iteration. Feasibility of \eqref{eq:dc} is dependent on the feasibility of the initial conditions.

We can show that \eqref{eq:quantile_reform_final_4} elicits a difference of convex formulation by subtracting $f_i(\vec{x}(0),\vec{U})$ from both sides. We also add slack variables to accommodate potentially infeasible initial conditions \cite{boyd_dc_2016,horst2000}. When using a difference of convex program, Lemma \ref{lem:2} guarantees a feasible but locally optimal solution.

\section{Experimental Results}\label{sec:results}

We demonstrate our method on three satellite rendezvous problems. All computations were done on a 1.80GHz i7 processor with 16GB of RAM, using MATLAB, CVX \cite{cvx} and Gurobi \cite{gurobi}. Polytopic construction and plotting was done with MPT3 \cite{MPT3}. All code is available at \url{https://github.com/unm-hscl/shawnpriore-t-dist-cwh}.

For all three scenarios, solution convergence between iterations was defined as the difference of sequential performance objectives and the sum of slack variables both less than $10^{-8}$. Difference of convex programs were limited to 100 iterations. The first order approximations of the reverse convex constraints were initially computed assuming no system input. We use the first four derivatives to compute the numerical approximation of the quantile functions. 

\subsection{Observational Mission} \label{ex:observe}

Consider a scenario in which two satellites are stationed in geosynchronous orbit. One satellite, the chief, has malfunctioning thrusters believed to be caused by a leaky seal in the thruster cap. The second satellite, the deputy, has been tasked with observing the chief to verify the cause of the malfunction. The deputy must navigate to desired locations around the chief and maintain an attitude orientation such that the chief is within the sensors $40^\circ$ field of view. Further, the deputy must avoid colliding with the chief at all time steps. We presume all motion is done in-plane. The relative translational dynamics of the deputy, with respect to the chief, are described by the planar Clohessy-Wilthire-Hill (CWH) equations \cite{wiesel1989_spaceflight}. We further presume the the angular momentum vector of the observing satellite is perpendicular to the plane and all torques applied are parallel to the plane. For Euler angle representations, this corresponds to rotations in the yaw parameter (denoted as $\theta$). To accommodate an additive $t$ disturbance to this parameter, we must assume $\theta$ can take any real value. The small angle approximation of the attitude kinematics allow us to model the attitude acceleration being equal to the applied torque of the satellites attitude control mechanism. Hence, the system dynamics are 
\begin{subequations} \label{eq:dynamics_observe}
\begin{align}
\ddot x - 3 \omega^2 x - 2 \omega \dot y = & \;  \frac{F_x}{m_c} \\
\ddot y + 2 \omega \dot x = & \;  \frac{F_y}{m_c} \\
\ddot \theta = & \; \frac{F_{\theta}}{J_{\theta}}
\end{align}   
\label{eq:cwh1}
\end{subequations}
with input $\vec{u} = [ \begin{array}{ccc} F_x & F_y & F_{\theta} \end{array}]^\top$ and $\omega = \sqrt{\frac{\mu}{R^3_0}}$. 

We discretize \eqref{eq:cwh1} under the assumption of impulse control, with sampling time $300$s, and insert a disturbance process that captures uncertainties in the model specification with respect to the chiefs malfunctioning thrusters, so that dynamics the deputy are described by  
\begin{equation}
    \bvec{x}(k+1) = \overline{A} \bvec{x}(k) + \overline{B} \vec{U}(k) + \bvec{w}(k) 
\end{equation}
with $\mathscr{U} = [-3m,3m]^2  \times [-90^\circ,90^\circ]$ per time step, and time horizon $N=8$, corresponding to 40 minutes of operation. We assume 
\begin{equation}
    \bvec{w}(k) \sim t(0, \overline{\Sigma}, 4)
\end{equation}
where $\overline{\Sigma} = \mathrm{diag}\left(10^{-4} \overline{I}_2, 10^{-6}, 5 \!\times\! 10^{-8} \overline{I}_2, 5 \!\times\! 10^{-10} \right)$ and that the dependence structure of the disturbances aligns with Assumption \ref{assm:distdepend}.  Here, the use of the multivariate $t$ is used as malfunctioning thrusters are likely to have a high propensity for outliers in comparison to a Gaussian distribution. We choose $\nu=4$ as for this value $\cov(\bvec{w}_i(k)) = 2\overline{\Sigma}$ implying high variance.

The terminal sets $\mathscr{T}_{k}$ are $4\times 4$m boxes centered approximately 10m away from the origin starting on the positive $x$ axis and progressing clockwise with two time steps in between each target set. Each target set requires the relative translational velocity be bounded in both directions by $[-0.01, 0.01]$m/s.  The attitude is bounded by $[-20^\circ,20^\circ]$ from $135^\circ$, $180^\circ$, $225^\circ$, and $270^\circ$, for time steps 2, 4, 6, and 8, respectively. This is to ensure the deputy has the chief within its field of view from any location in the translational target sets.  The rotational velocity is bounded by $[-0.01, 0.01]$deg/s. For collision avoidance, we presume that the deputy must remain at least $r=8$m away from the chief, hence $\overline{S} = \begin{bmatrix} \overline{I}_{2} & \overline{0}_{2 \times 4} \end{bmatrix}$ to extract the positions.  Violation thresholds for terminal sets and collision avoidance are $\alpha_{\mathscr{T}} = \alpha_o = 0.2$, respectively.
\begin{align}
    \Prob \left\{ \bigcap_{k=1}^4 \bvec{x}(2k) \in \mathscr{T}_{2k} \right\} &\geq 1-\alpha_{\mathscr{T}} \label{eq:terminal_observe}\\
    \Prob \left\{ \bigcap_{k=1}^{8} \left\| \overline{S} \!\cdot\! \bvec{x}(k) \right\| \geq r \right\} &\geq 1-\alpha_o  \label{eq:avoidance_observe}
\end{align}

The performance objective is based on fuel consumption. 
\begin{equation}
    J(\vec{U}) = \vec{U}^\top \vec{U}
\end{equation}

\begin{figure}
    \centering
    \includegraphics[width=0.85\columnwidth]{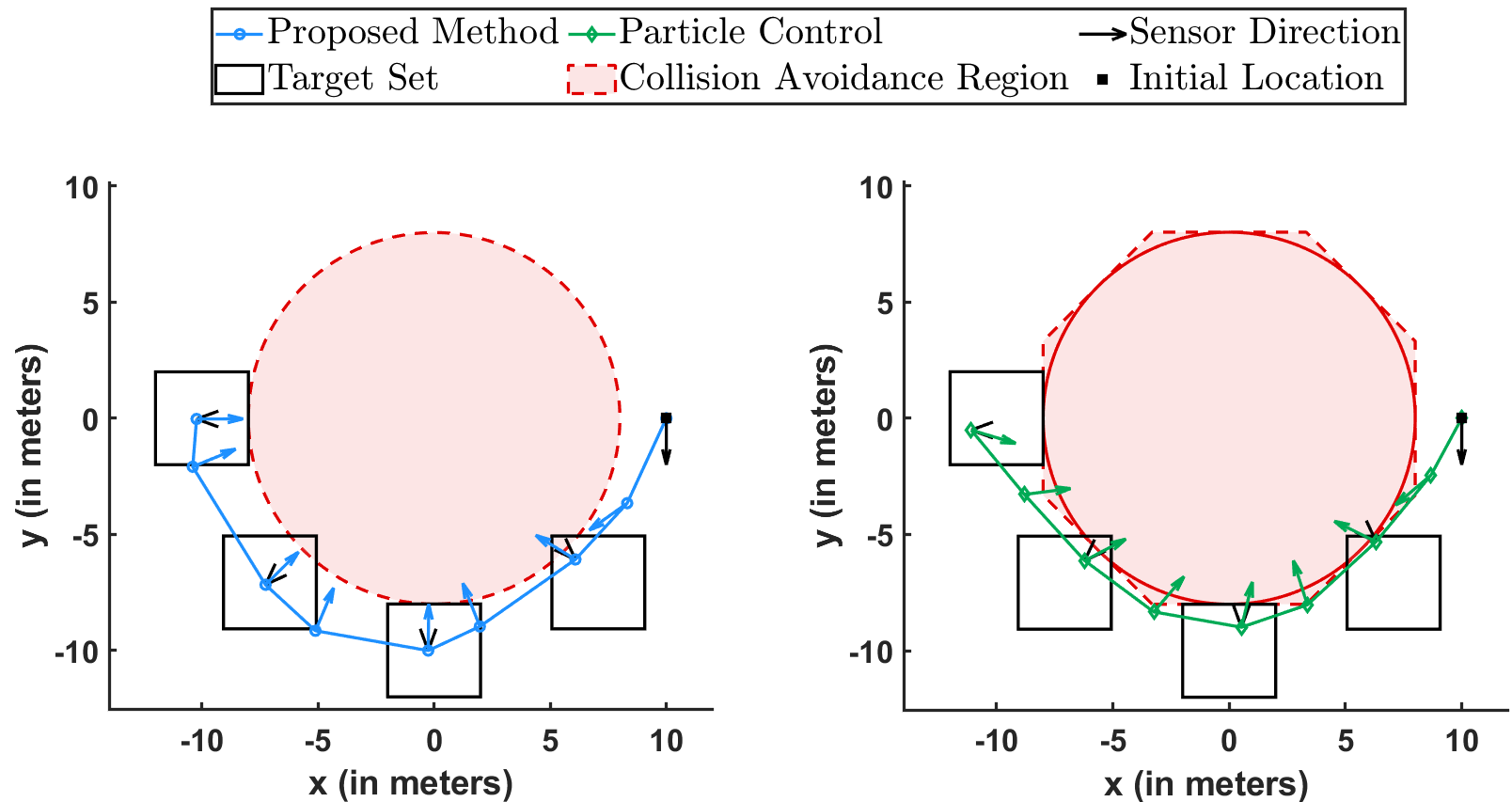}
    \caption{Comparison of trajectories between proposed method (top) and particle control (bottom) in CWH frame. Here, the black 'v' extending from the satellite's position is the attitude target set, and the arrows present the attitude of the spacecraft. The octagonal region is a linear over-approximation of the 2-norm used in our method. The black arrow and square of the right hand side of each graph represent the initial translational and attitude conditions of the satellite.}
    \label{fig:cwh_traj_compare}
\end{figure}

\begin{table}
    \caption{Comparison of Solution and Computation Time for dynamics \eqref{eq:dynamics_observe} with Multivariate-$t$ Disturbance.}
    \centering
    \begin{tabular}{lcc}
         \hline \hline
         Metric &  Proposed Method &  Particle Control \\
         \hline 
         Computation Time & 4.0086 sec & 301.6698 sec\\ 
         Iterations to Converge &  7 & N/a \\ 
         Objective Cost for Derived Solution & $8.55 \times 10^{-4}$ & $2.11 \times 10^{-4}$ \\ \hline
    \end{tabular}
    \label{tab:cwh_stats_observe}
\end{table}

When approximating the numerical quantiles, we presume intervals $h=5\times 10^{-6}$, and maximum approximation error $\xi = 0.01$. For the Student's $t$ distributions, we set the instantiating point, $p_0$, to $0.5$ with known quantile $\Phi^{-1}(p_0) = 0$. For the beta prime distributions, we set the instantiating point to $0.5$. Since, $\gamma = 1$, computation of $\Phi^{-1}(p_0)$ was completed using the square root of the analytical median expression \eqref{eq:beta_prime_med_a}. Each quantile approximation used the first four derivatives of the pdf. We note that Assumption \ref{assm:convex} is met as $\gamma = 1$ implies the quantile is convex over $p=[0,1]$.

\begin{table}
    \caption{Constraint Satisfaction for CWH dynamics with Multivariate-$t$ Disturbance, with $10^4$ Samples and Probabilistic Violation Threshold of $1-\alpha =0.8$. Satisfaction 'SAT' of the constraint is marked with a $\checkmark$.}
    \centering
    \begin{tabular}{lcccc}
         \hline 
         \hline
         Constraint &  Proposed Method &  SAT & Particle Control & SAT \\ \hline 
         Terminal Set \eqref{eq:terminal_observe} & 0.8721 & $\checkmark$ & 0.5596 & \\ 
         Collision Avoidance with Chief \eqref{eq:avoidance_observe} & 0.9726 & $\checkmark$ &  0.8243 & $\checkmark$ \\
         \hline
    \end{tabular}
    \label{tab:cwh_constraint_observe}
\end{table}

We compare the proposed method with the mixed integer particle approach using a the mixed integer linear program (MILP) \cite{ono2008iterative}. To facilitate comparison, we only need to modify the collision avoidance constraint. As the circular region defined by the 2-norm cannot be encoded into a MILP, we use an 8-sided polytope to overapproximate the 2-norm collision avoidance region. We generated 25 disturbance sequences to generate an open loop controller. The resulting trajectories, costs, and computation times differ drastically as shown in Figure \ref{fig:cwh_traj_compare} and Table \ref{tab:cwh_stats_observe}. To assess constraint satisfaction, we generated $10^4$ Monte Carlo sample disturbances for each approach. Table \ref{tab:cwh_constraint_observe} shows that while both methods satisfied the collision avoidance constraint, particle control did not satisfy the safety threshold for the target set constraint. 

The proposed method performed two to three orders of magnitude faster than particle control. Given the significant increase in binary variables needed to perform particle control, this comes as no surprise. Conversely, the low number of disturbance samples is likely the cause for the poor performance with respect to the target set constraint. Given the random nature of the sampling process, 25 samples is not enough to characterize the behaviour on a larger scale, particularly for a heavy tailed distribution. We attempted to use 50 disturbance samples but could not find a solution in under 3 hours.

The over approximation of the 2-norm collision avoidance is likely the reason the particle control solution trajectory satisfied the collision avoidance constraint. In the corners of the octagonal region the collision avoidance constraint is nearly a meter larger than needed. This forces the solution further away from collision avoidance region then it would have given a more accurate approximation of the region. We attempted to perform the particle control approach with polytopic collision avoidance regions defined with more half-space constraints to closer approximate the 2-norm region. In an attempt to use 16 sides, the method could not find a solution in under 24 hours. 

The objective cost for our method is nearly 4 times larger than that of particle control. We know that the true optimal trajectory has a cost lying somewhere between the two as one is to conservative and the other doesn't meet the safety thresholds. Given the large differences between the specified probabilistic safety threshold and the sampled constraint satisfaction, we can say our method has introduced significant conservatism. This comes as no surprise given the use of Boole's inequality for both constraints, and the use of the reverse triangle and parallelogram laws for the reformulation of the collision avoidance constraint. 

There is, however, one benefit to the introduced conservatism. In this demonstration, we were able to use analytical results to establish instantiation points for our quantile approximation. This, in conjunction with the conservatism, should near guarantee constraint satisfaction for Problem \ref{prob:2} as a result of Lemma \ref{lem:3}. While we cannot prove this to be the case, empirical results are likely to reflect this conclusion under similar conditions.  

\subsection{High Capacity Docking} \label{ex:dock}

Consider a scenario in which seven satellites are stationed in geosynchronous orbit. Each satellite is tasked with reaching a terminal target set representing a docking location with a static refueling station. Each satellite must avoid other satellites and the refueling station while navigating to their respective target sets. The relative dynamics of each spacecraft, with respect to the known location of the refueling station, are described by the CWH equations
\cite{wiesel1989_spaceflight}
\begin{subequations}
\begin{align}
\ddot x - 3 \omega^2 x - 2 \omega \dot y &= \frac{F_x}{m_c} \label{eq:cwh:a}\\
\ddot y + 2 \omega \dot x & = \frac{F_y}{m_c} \label{eq:cwh:b}\\
\ddot z + \omega^2 z & = \frac{F_z}{m_c}. \label{eq:cwh:c}
\end{align}   
\label{eq:cwh}
\end{subequations}
with input $\vec{u}_i = [ \begin{array}{ccc} F_x & F_y & F_z\end{array}]^\top$, and $\omega = \sqrt{\frac{\mu}{R^3_0}}$. 

We discretize \eqref{eq:cwh} under the assumption of impulse control, with sampling time $300$s, and insert a disturbance process that captures model uncertainties and uncaptured perturbations, so that dynamics for vehicle $i$ are described by  
\begin{equation}
    \bvec{x}_i(k+1) = \overline{A} \bvec{x}_i(k) + \overline{B} \vec{U}_i(k) + \bvec{w}_i(k) \label{eq:cwh_lin}
\end{equation}
with $\mathscr{U}_i = [-3,3]^3$, and time horizon $N=8$, corresponding to 40 minutes of operation. We assume 
\begin{equation}
    \bvec{w}_i(k) \sim t(0, \overline{\Sigma}, 20)
\end{equation}
where $\overline{\Sigma} = \mathrm{diag}\left(10^{-4} \overline{I}_3, 5 \!\times\! 10^{-8} \overline{I}_3 \right)$ and that the dependence structure of the disturbances aligns with Assumption \ref{assm:distdepend}. Here, the use of the multivariate $t$ is used to model perturbation forces of interest but not captured in the CWH dynamics. This includes drag, solar radiation pressure, 3\textsuperscript{rd} body acceleration from the Sun and Moon, and impacts with small but unknown debris. We choose $\nu=20$ as the combined effect of these perturbation forces may be small but are likely outliers in comparison to a Gaussian distribution.

The terminal sets $\mathscr{T}_{iN}$ are $5\times 5 \times 5$m boxes centered around desired terminal locations in $x,y,z$ coordinates approximately 9m away from the origin, with velocity bounded in all three directions by $[-0.01, 0.01]$m/s.  For collision avoidance, we presume that all satellites must remain at least $r=8$m away from each other and the refueling station, hence $\overline{S} = \begin{bmatrix} \overline{I}_{3} & \overline{0}_{3\times 3} \end{bmatrix}$ to extract the positions. We presume the collision avoidance constraints are valid only for the non-terminal time steps.  Violation thresholds for terminal sets and collision avoidance are $\alpha_{\mathscr{T}} = \alpha_r = \alpha_o = 0.2$, respectively.
\begin{align}
    \Prob \left\{ \bigcap_{i=1}^7 \bvec{x}_i(N) \in \mathscr{T}_{iN} \right\} &\geq 1-\alpha_{\mathscr{T}} \label{eq:terminal}\\
    \Prob \left\{ \bigcap_{k=1}^{7} \bigcap_{i=1}^{7} \left\| \overline{S} \!\cdot\! \bvec{x}_i(k) \right\| \geq r \right\} &\geq 1-\alpha_o  \label{eq:avoidance_origin}\\
    \Prob \left\{ \bigcap_{k=1}^{7} \bigcap_{i,j=1}^{7} \left\| \overline{S} \!\cdot\!\left(\bvec{x}_i(k)\! - \! \bvec{x}_j(k)\right)\right\| \geq r \right\} &\geq 1-\alpha_r  \label{eq:avoidance}
\end{align}
We note that \eqref{eq:avoidance_origin}-\eqref{eq:avoidance} has a combined 196 collision avoidance constraints to be embedded in the problem.

The performance objective is based on fuel consumption. 
\begin{equation}
    J(\vec{U}_1, \dots, \vec{U}_7) = \sum^7_{i=1} \vec{U}_i^\top \vec{U}_i
\end{equation}

When approximating the numerical quantiles, we presume intervals $h=5\times 10^{-6}$, and maximum approximation error $\xi = 0.01$. For the Student's $t$ distributions, we set the instantiating point, $p_0$, to $0.5$ with known quantile $\Phi^{-1}(p_0) = 0$. For the beta prime distributions, we set the instantiating point to $0.5$. Computation of $\Phi^{-1}(p_0)$ was completed using the median approximation \cite{Lyon2021}
\begin{equation} \label{eq:lyon_median}
    \Phi^{-1}(0.5) = \sqrt{\frac{2^{-\frac{1}{\gamma}}(\log(2)-\frac{1}{2}+\gamma)}{2^{-\frac{1}{\delta}}(\log(2)-\frac{1}{2}+\delta)} }   
\end{equation}
Each quantile approximation used the first four derivatives of the pdf. This median approximation was compared against the mode of the distribution to verify that $\alpha_r$ and $\alpha_o$ were in the convex region of the quantile. 

\begin{figure}
    \centering
    \includegraphics[width=0.85\columnwidth]{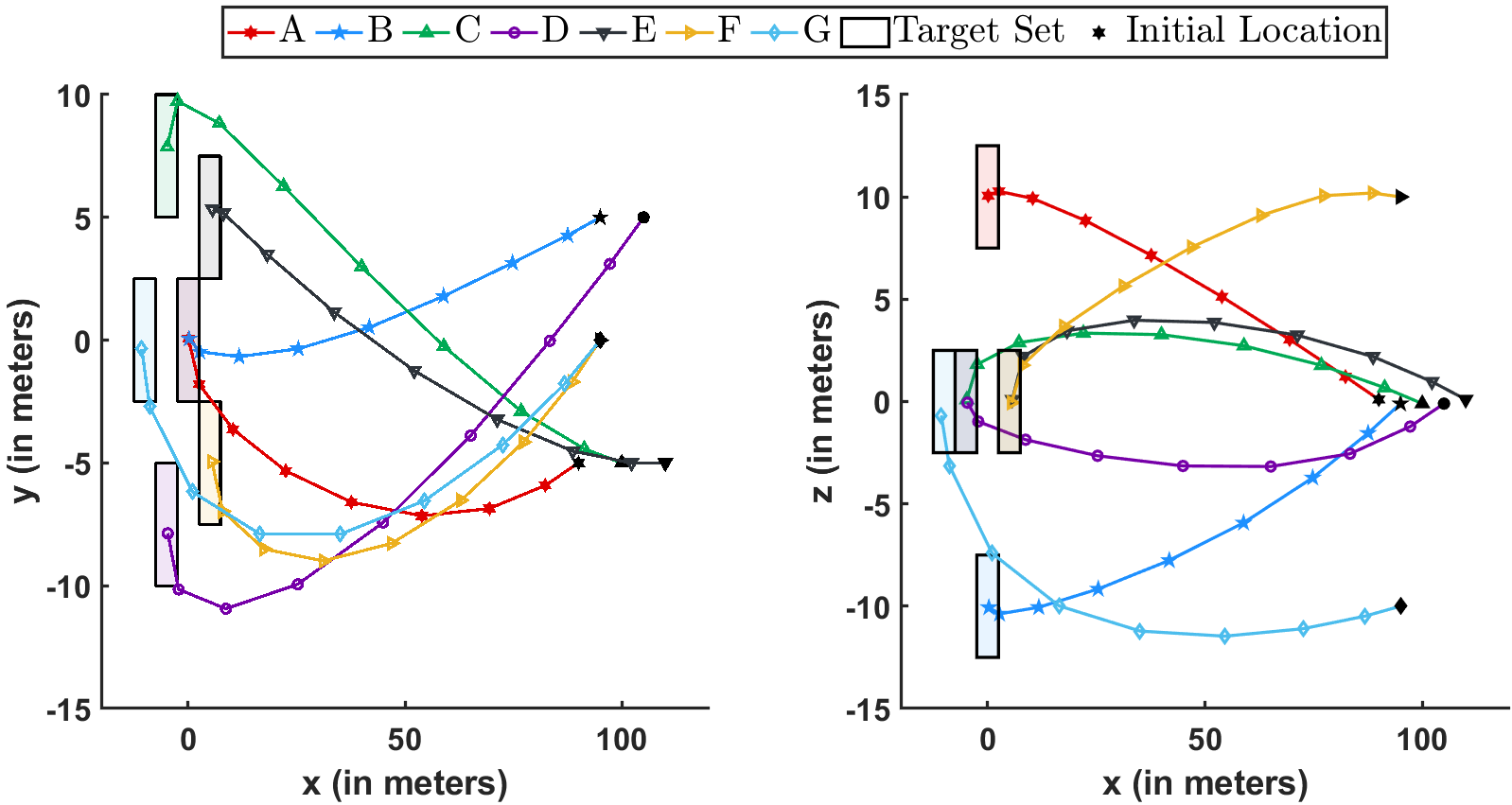}
    \caption{Trajectories of the seven satellites in CWH frame. }
    \label{fig:cwh_traj}
\end{figure}

\begin{table}
    \caption{Constraint Satisfaction for CWH dynamics with Multivariate-$t$ Disturbance, with $10^4$ Samples and Probabilistic Violation Threshold of $1-\alpha =0.8$. Satisfaction 'SAT' of the constraint is marked with a $\checkmark$.}
    \centering
    \begin{tabular}{lcc}
         \hline\hline
         Constraint &  Sample Satisfaction & SAT\\
         \hline 
         Terminal Set \eqref{eq:terminal} & 0.8207 & $\checkmark$\\ 
         Collision Avoidance with Refueling Station \eqref{eq:avoidance_origin} & 0.9872& $\checkmark$ \\ 
         Inter-satellite Collision Avoidance \eqref{eq:avoidance} &  0.9869 & $\checkmark$\\
         \hline
    \end{tabular}
    \label{tab:cwh_constraint}
\end{table}

\begin{table}
    \caption{Computation Statistics for CWH dynamics with Multivariate-$t$ Disturbance.}
    \centering
    \begin{tabular}{lc}
         \hline\hline
         Metric &  Value \\
         \hline 
         Computation Time to Solve Problem 2& 33.3759  sec \\
         Total Computation Time & 43.3393 sec \\ 
         Iterations to Converge &  34 \\ 
         Objective Cost for Derived Solution & 0.015873 \\ \hline
    \end{tabular}
    \label{tab:cwh_stats}
\end{table}

The resulting trajectories are shown in Figure \ref{fig:cwh_traj}. To assess constraint satisfaction, we generated $10^4$ Monte-Carlo sample disturbances for each approach. Table \ref{tab:cwh_constraint} shows that all constraints were satisfied to the required threshold. Note that all three constraints are satisfied to a more conservative threshold implying that the reformulation \eqref{eq:quantile_final} is not a tight upper bound for the problem. These results are consistent with the finding in Section \ref{ex:observe}. We note that solutions using particle control could not be found within a weeks time for 25 disturbance samples; thus, we do not provide a comparison for this example.

Table \ref{tab:cwh_stats} provides computational statistics on the difference of convex program. The proposed method computed the trajectories in under a minute. With nearly 200 collision avoidance constraints embedded in this problem, solution convergence of this speed warrants further consideration for this method. 

\subsection{Monte Carlo Simulation}

Consider a scenario in which three satellites are stationed in geosynchronous orbit. The satellites have been caught in an unexpected small debris field caused by a nearby collision. The satellites must cooperatively reach a new configuration outside of the debris field. Each satellite must reach a desired target set while avoiding collision with the other satellites. We again use the CWH equations \eqref{eq:cwh}. We discretize \eqref{eq:cwh} under the assumption of impulse control, with sampling time $300$s, and insert a disturbance process, so that dynamics for vehicle $i$ are described by  
\begin{equation}
    \bvec{x}_i(k+1) = \overline{A} \bvec{x}_i(k) + \overline{B} \vec{U}_i(k) + \bvec{w}_i(k)
\end{equation}
with $\mathscr{U}_i = [-3,3]^3$, and time horizon $N=8$, corresponding to 40 minutes of operation. We assume 
\begin{equation}
    \bvec{w}_i(k) \sim t(0, \overline{\Sigma}, 4)
\end{equation}
where $\overline{\Sigma} = \mathrm{diag}\left(10^{-4} \overline{I}_3, 5 \!\times\! 10^{-8} \overline{I}_3 \right)$ and that the dependence structure of the disturbances aligns with Assumption \ref{assm:distdepend}. Here, the use of the multivariate $t$ is used to model perturbation forces of the small debris colliding with the satellite. 

The terminal sets $\mathscr{T}_{iN}$ are $6\times 6 \times 6$m boxes centered around desired terminal locations in $x,y,z$ coordinates approximately 9m away from the origin, with velocity bounded in all three directions by $[-0.1, 0.1]$m/s.  For collision avoidance, we presume that all satellites must remain at least $r=8$m away from each other, hence $\overline{S} = \begin{bmatrix} \overline{I}_{3} & \overline{0}_{3\times 3} \end{bmatrix}$ to extract the positions. We presume the collision avoidance constraints are valid for all time steps. Violation thresholds for terminal sets and collision avoidance are $\alpha_{\mathscr{T}} = \alpha_r = 0.2$, respectively.
\begin{align}
    \Prob \left\{ \bigcap_{i=1}^3 \bvec{x}_i(N) \in \mathscr{T}_{iN} \right\} &\geq 1-\alpha_{\mathscr{T}} \label{eq:terminal_mc}\\
    \Prob \left\{ \bigcap_{k=1}^{8} \bigcap_{i,j=1}^{3} \left\| \overline{S} \!\cdot\!\left(\bvec{x}_i(k)\! - \! \bvec{x}_j(k)\right)\right\| \geq r \right\} &\geq 1-\alpha_r  \label{eq:avoidance_mc}
\end{align}

The performance objective is based on fuel consumption. 
\begin{equation}
    J(\vec{U}_1, \dots, \vec{U}_3) = \sum^3_{i=1} \vec{U}_i^\top \vec{U}_i
\end{equation}

When approximating the numerical quantiles, we presume the same methodology prescribed in Section \ref{ex:dock}. 

We demonstrate our method with 1,000 Monte Carlo disturbances added to our initial condition's position elements. Here, the disturbances are sampled independently from a multivariate $t$ distribution with zero mean, identity correlation matrix, and 10 degrees of freedom. Each set of random initial conditions were used to solve Problem \ref{prob:2}. We note that all 1,000 sets of initial conditions generated a solution in less than the maximum allowable iterations of 100. A descriptive plot of the resulting trajectories are shown in Figure \ref{fig:mc_traj}. Here, the error bars represent the minimum and maximum of derived solutions path at each time step. The error bars are centered on the average of the set of derived solutions. We see that even though the initial conditions form a high variance distribution, the optimized trajectories converge over time. By the terminal step, we see the distribution of the trajectories being approximately over 1m$\times$1m$\times$1m sets.  

\begin{figure}
    \centering
    \includegraphics[width=0.85\columnwidth]{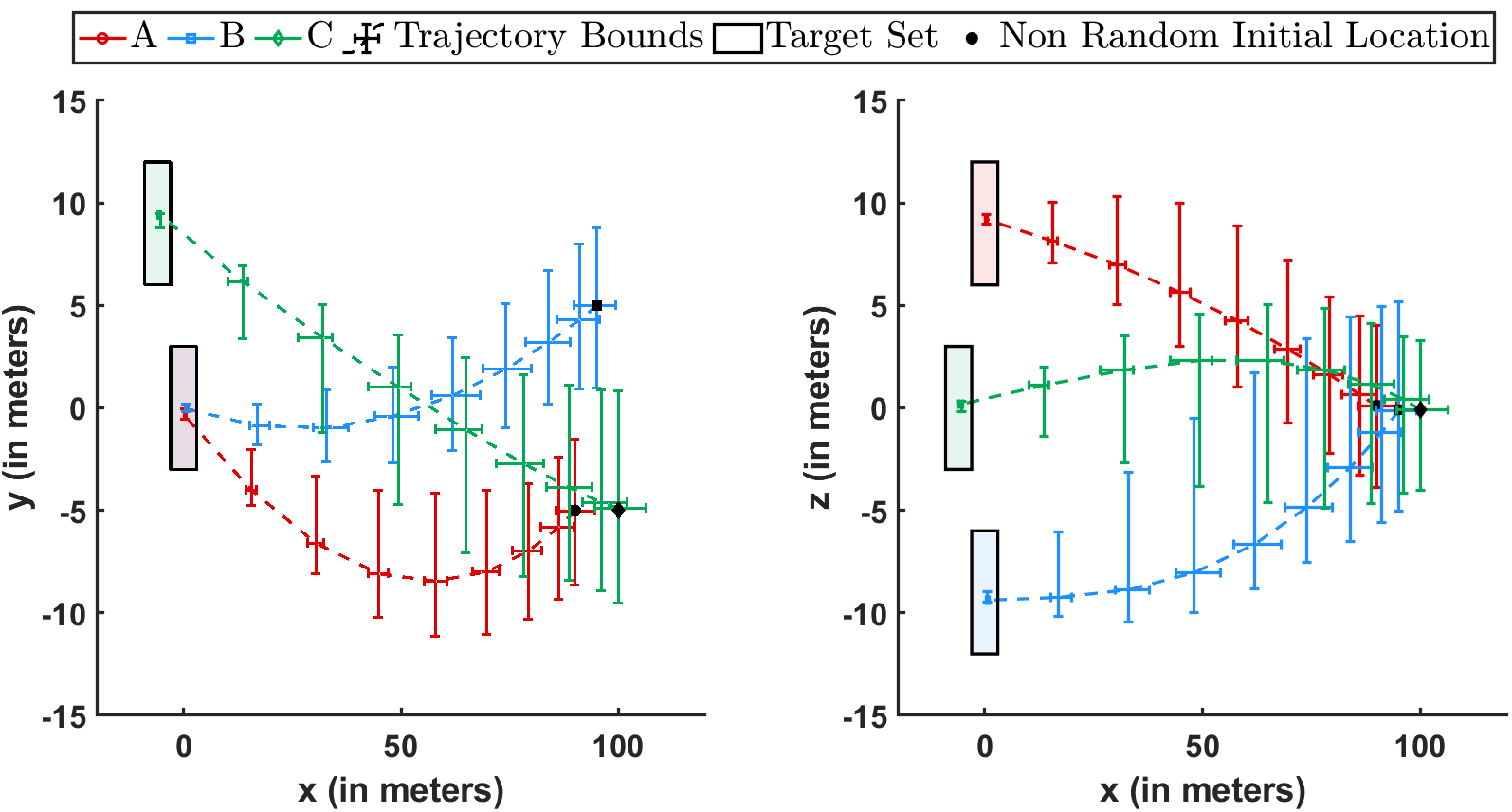}
    \caption{Descriptive plot of the trajectories of the three satellites in CWH frame with 1,000 Monte Carlo generated disturbances added to the non-random initial conditions. Here, the error bars represent the minimum and maximum values of the derived trajectories at each time step and are centered on the average path.}
    \label{fig:mc_traj}
\end{figure}

\begin{table} 
    \caption{Descriptive Statistics for CWH dynamics with Multivariate-$t$ Disturbance and 1,000 Random Initial Conditions.}
    \centering
    \begin{tabular}{lcccc}
    \hline \hline
         Metric &  Mean & Std Dev & Min & Max\\ \hline
         Computation time (sec) & 4.8043 & 0.9437 & 2.8522 & 18.0091\\
         Optimal Cost ($N^2$) & $1.7426\times 10^{-3}$ & $2.9013\times 10^{-5}$ & $1.6555\times 10^{-3}$ & $1.8897\times 10^{-3}$ \\
         Iterations to Converge & 12.2270 & 2.3330 & 8 & 45 \\ \hline
    \end{tabular}
    \label{tab:mc_decrpt}
\end{table}

\begin{table}
    \caption{Constraint Satisfaction Descriptive Statistics for CWH dynamics with Multivariate-$t$ Disturbance and 1,000 Random Initial Conditions. Constraint Satisfaction Computed with $10^4$ Samples and Probabilistic Violation Threshold of $1-\alpha =0.8$.}
    \centering
    \begin{tabular}{lcccc}
    \hline \hline
         Constraint &  Mean & Std Dev & Min & Max\\ \hline
         Terminal Set \eqref{eq:terminal_mc} & 0.8418 & $3.6536\times 10^{-3}$ & 0.8294 & 0.8518 \\
         Inter-satellite Collision Avoidance \eqref{eq:avoidance_mc} & 0.9941 & $8.2951 \times 10^{-4}$ & 0.9911 & 0.9967 \\ \hline
    \end{tabular}
    \label{tab:mc_constraint}
\end{table}

Tables \ref{tab:mc_decrpt} and \ref{tab:mc_constraint} provide descriptive statistics for the 1,000 solutions. We see that the initial conditions did not have much effect on the computation time and optimal costs, as evidenced by the small standard deviations. Of note, the maximum values listed in Table \ref{tab:mc_decrpt} are all from the same run. In this instance, the three vehicles had initial conditions within 2m of each other requiring a larger initial change in velocity to meet the collision avoidance requirement. This led to an increase in iterations as a result of the added slack variables only being lightly penalized in the early iterations. 

In Table \ref{tab:mc_constraint} we see that all trajectories satisfied the probabilistic constraints to the required threshold. Of interest is the spread of the two constraints over the 1,000 trajectories. We see that the standard deviation for the terminal constraint \eqref{eq:terminal_mc} is larger than that of the collision avoidance constraint \eqref{eq:avoidance_mc} by an order of magnitude. This is likely caused by the difference in the levels of conservatism introduced by the quantile reformulations. Since the reverse convex collision avoidance constraint's reformulation adds significant conservatism, the range of values the true probability can take is more restricted. Thus, creating this apparent difference. 

From Figure \ref{fig:mc_traj} and Tables \ref{tab:mc_decrpt} and \ref{tab:mc_constraint}, it is apparent that the method is not sensitive to the initial conditions. This comes as no surprise as the method is designed to optimize based on constraints from the first time step forward. This implies that so long as the control authority is sufficient that the satellites can reach the required distance with the first input sequence, a solution is likely to exist and can be found. 

\section{Conclusion}\label{sec:conclusion}

We proposed a framework for synthesising stochastic optimal controllers for LTI systems under heavy tailed disturbances modeled with the multivariate $t$ distribution. Our approach relies on a  affine numerical approximations of unknown quantile functions via a Taylor series expansion. We embed the affine quantile approximation in a quadratic difference-of-convex programs that solves a conservative reformulation of the original problem. We demonstrated our approach on three satellite rendezvous scenarios with varying system parameterizations and safety requirements. Our results show that the proposed method is not only computationally efficient but also adaptable to many scenarios. 

\section*{Appendix} \label{appx:prop}
\begin{prop}[Marginal distributions of a sub-vectors of a multivariate $t$ random vector \cite{Sutradhar1986}] \label{prop:t_margin}\hfill \\
Suppose $\bvec{x} \in \R^n \sim t \left(\vec{\mu}, \overline{\Sigma}, \nu \right)$ where
\begin{subequations}
\begin{align}
    \bvec{x} = & \begin{bmatrix} \bvec{x}_1 & \bvec{x}_2 \end{bmatrix}^{\top} \\
    \vec{\mu} = & \begin{bmatrix} \vec{\mu}_1 & \vec{\mu}_2 \end{bmatrix}^{\top} \\
    \overline{\Sigma} = & \begin{bmatrix} \overline{\Sigma}_{11} & \overline{\Sigma}_{12} \\ \overline{\Sigma}_{21} & \overline{\Sigma}_{22} \end{bmatrix}
\end{align}
\end{subequations}
The partitioned vectors $\bvec{x}_1 \in \R^m$ and $\bvec{x}_2 \in \R^{n-m}$ with partitioned location and scale parameters, $\vec{\mu}$ and $\overline{\Sigma}$, respectively, have the marginal distributions,
\begin{subequations}
\begin{align}
    \bvec{x}_1 \sim & t \left(\vec{\mu}_1, \overline{\Sigma}_{11}, \nu \right)\\
    \bvec{x}_2 \sim & t \left(\vec{\mu}_2, \overline{\Sigma}_{22}, \nu \right)
\end{align}
\end{subequations}
\end{prop}
\begin{prop}[Affine transformations of a multivariate $t$ random vector \cite{Sutradhar1986}] \label{prop:t_affine} \hfill \\
Suppose $\bvec{x} \in \R^n \sim t \left(\vec{\mu}, \overline{\Sigma}, \nu \right)$. For $\overline{B} \in \R^{m \times n}$ with full row rank and  $\vec{b} \in \R^m$ with $m \leq n$ then
\begin{equation}
    \overline{B} \bvec{x} + \vec{b} \sim t(\overline{B} \vec{\mu} + \vec{b}, \overline{B} \overline{\Sigma} \overline{B}^{\top}, \nu)
\end{equation}
\end{prop}

\begin{prop}[Norm of Standard Multivariate $t$ r.v.] \label{prop:beta_prime_norm} \hfill \\
Suppose $\bvec{x} \in \R^n \sim t \left(\vec{0}, \overline{I}_n, \nu \right)$. By construction, $\bvec{x} \equiv \frac{\bvec{y}}{\sqrt{\boldsymbol{z}/\nu}}$. Then
\begin{equation}
    \|\bvec{x}\|^2 \equiv \frac{\|\bvec{y}\|^2}{\boldsymbol{z}/\nu}
\end{equation}
where $\|\bvec{y}\|^2 \sim CSquare(n)$. Note that the chi-square distribution is a special case of the gamma distribution with shape $\frac{p}{2}$ and scale parameter $2$.  By dividing by $\nu$, we get $\frac{1}{\nu}\|\bvec{x}\|^2 \sim BPrime \left(\frac{n}{2}, \frac{\nu}{2}\right)$.
\end{prop}

\begin{prop}[Summation of $n$ beta prime r.v.'s (infinite divisibility property)  \cite{steutel_harn_2003}] \label{prop:beta_prime_sum} \hfill \\
Suppose $\bvec{x}_1, \dots, \bvec{x}_n \overset{i.i.d.}{\sim} BPrime (\gamma, \delta)$. Then
\begin{subequations}
\begin{align}
    \sum_{i=1}^n \bvec{x}_i \sim& \; BPrime (\phi, \psi) \\
    \phi = & \; \frac {n\gamma (\gamma +\delta ^{2}-2\delta +n\gamma \delta -2n\gamma +1)}{(\delta -1)(\gamma +\delta -1)} \\
    \psi = & \; \frac {2\gamma +\delta ^{2}-\delta +n\gamma \delta -2n\gamma }{\gamma +\delta -1}
\end{align}
\end{subequations}
\end{prop}

\section*{Funding Sources}

This material is based upon work supported by the National Science Foundation under NSF Grant Number CMMI-2105631.  Any opinions, findings, and conclusions or recommendations expressed in this material are those of the authors and do not necessarily reflect the views of the National Science Foundation.  

\bibliography{main}
\end{document}